\documentclass[preprint,12pt]{article}
  \usepackage{graphicx}
  \usepackage{amsfonts}
  \usepackage[english]{babel}
  \usepackage{authblk}
  
\newtheorem{theorem}{Theorem}[section]
\newtheorem{lemma}[theorem]{Lemma}

\newenvironment{proof}[1][Proof]{\begin{trivlist}
\item[\hskip \labelsep {\bfseries #1}]}{\end{trivlist}}

\makeatletter

\begin{document}

  \title{Analytic solution of a model of language competition with bilingualism and 
          interlinguistic similarity}

  \author[1]{M. V. Otero-Espinar$^*$}
  \author[2,3]{L. F. Seoane$^*$}
  \author[1,4]{J. J. Nieto}
  \author[5]{J. Mira$^\dagger$}

  \affil[1]{Departamento de An\'alise Matem\'atica and Instituto de Matem\'aticas, 
                  Universidade de Santiago de Compostela, 15782 Santiago de Compostela, Spain}

  \affil[2]{ICREA-Complex Systems Lab, Universitat Pompeu Fabra, Dr Aiguader 88, 08003 
                  Barcelona, Spain.}

  \affil[3]{Institut de Biologia Evolutiva, UPF-CSIC, Psg Barceloneta 37, 08003 Barcelona, 
        Spain.}

  \affil[4]{Department of Mathematics, Faculty of Science, King Abdulaziz University, Jeddah, 
                Saudi Arabia. }

  \affil[5]{Departamento de FÃ\'isica Aplicada, Universidade de Santiago de Compostela, 15782 
          Santiago de Compostela, Spain. 

          $^*$ These authors contributed equally to this work.

          $^\dagger$ Corresponding author: jorge.mira@usc.es}

  \date{\today}

  \maketitle
  
  \begin{abstract}
 
    An in-depth analytic study of a model of language dynamics is presented: a model which tackles
the problem of the coexistence of two languages within a closed community of speakers taking into
account bilingualism and incorporating a parameter to measure the distance between languages. After
previous numerical simulations, the model yielded that coexistence might lead to survival of both
languages within monolingual speakers along with a bilingual community or to extinction of the
weakest tongue depending on different parameters. In this paper, such study is closed with thorough
analytical calculations to settle the results in a robust way and previous results are refined with
some modifications. From the present analysis it is possible to almost completely assay the number
and nature of the equilibrium points of the model, which depend on its parameters, as well as to
build a phase space based on them. Also, we obtain conclusions on the way the languages evolve with
time. Our rigorous considerations also suggest ways to further improve the model and facilitate the
comparison of its consequences with those from other approaches or with real data.

  \end{abstract}

  \section{Introduction} 
    \label{sec:intro}

    At whatever scale that we look, languages reveal themselves as very elaborated entities
consisting of many coupled parts: grammar, vocabulary, etc; each of them complex in its own nature
as well. They are, moreover, a main instrument of interaction in an entangled web of social agents
so that the state and evolution of tongues cannot, ultimately, be considered as detached from other
social dynamics. We readily appreciate that we are in front of an utter challenge to the human
intellect \cite{MaynardSmith-Szathmary, Wray}. Small steps are gradually taken towards a further
understanding of the many problems posed by languages. Leaving aside those contributions from the
more classic fields (e.g. philology), linguistic questions were opened to very diverse branches of
science during the 20th century by drawing inspiration from some pioneer multi- disciplinary works
\cite{Hawkings-GellMann}. Given the complexity outlined before, any of these first transversal
approaches are necessarily simplistic or rely largely on computer simulations, and rigorous and
definitive mathematical proofs of the results are often missing. 

    The kind of questions that were exposed to a more varied community of researchers regard the
evolution of languages: transformations in their syntaxes, grammars, or vocabularies; aging, rise,
and death; the dynamics of their number of speakers: spreading of culture, competition or other kind
of interaction with other tongues; etc. And the fields that take on these issues are as diverse as
sociology, biology, or physics. The science of complex systems should be highlighted because of its
very clever usage of existing mathematical methods that stem mainly from statistical mechanics \cite
{Loreto-Steels}. While also appraising other important contributions that widen our knowledge about
the nature of human languages \cite{Axelrod, Dyen-Black, Petroni-Serva, Schulze-Wichmann,
CorominasMurtra-Sole, Sole-Steels, NelsonSathi-Dagan, Amancio-Costa}, we shall focus on a seminal
paper by Abrams and Strogatz \cite {Abrams-Strogatz} that prominently triggered research in its
direction. For an exhaustive review on very varied related topics with up-to-date bibliography
consult \cite{Sole-Fortuny}; and for a more extended review on the impact of statistical physics on
social dynamics, including language modeling, see \cite{Castellano-Loreto}.

    The line of research propelled by \cite{Abrams-Strogatz} addresses the modeling of language
coexistence as a competitive dynamics to attract speakers. In \cite{Abrams-Strogatz} a minimal
model was accounted for and in accordance with experimental data a sounded result spread: that a
two-languages competition for speakers always led to the extinction of one of the parties. Further
analysis of the model \cite{Stauffer-SanMiguel, Chapel-SanMiguel} shows that it also allows for
stable language coexistence, but the parametric setup needed has not been observed in any study
with available real data. Following the trend, more complicated models were developed that took
spatial or social structure into account \cite{Patriarca-Leppanen, Patriarca-Heinsalu, 
Vazquez-SanMiguel} or that explicitly introduced bilingualism \cite{Vazquez-SanMiguel, Wang-Minett,
Castello-SanMiguel, Mira-Paredes, Minett-Wang, Mira-Nieto, Patriarca-SanMiguel}. This naturally
eased the way to solutions with stable coexisting languages.

    As it was advanced before, computer simulations constitute a favorite tool in this modern wave
of scientific approaches to the study of languages. The results are usually convincing more than
enough and, besides, these numerical studies allow to reach a depth of knowledge that might be
impossible if we should rely only on very rigorous analytical demonstrations. Despite of this, many
of the most insightful contributions to the comprehension of human communication follow from
meticulous and carefully proven mathematical constructions, mainly in the study of grammars and
largely aided by methods from computational sciences \cite{Nowak-Niyogi}.

    In this paper we intend to make a contribution by analytically elucidating some existing results
in the modeling of language competition. This is necessarily a rearguard job--since computer
simulations have the lead by far--but we will see how it is very valuable and necessary. Thanks to
the thorough reasoning of this paper we gain a deep understanding about the dynamics of speakers of
coexisting languages. We focus on a model of language competition that allows bilingualism
introduced in \cite{Mira-Paredes} and whose most interesting results were numerically derived in
\cite{Mira-Nieto}. By analyzing the model we will come to a better interpretation of the previous
numerical work and we will reach some new results that are, now, supported by robust analytical
proofs. \\

    The paper is structured in the following way: In section \ref{sec:derivation} the hypothesis
used in \cite{Mira-Paredes} are formulated and the corresponding equations are derived therefrom.
In section \ref{sec:resolution} strict mathematical results are carefully obtained. Whenever the
analytic tools do not reach to fully solve the problem, numerical simulations are employed, but its
presence--usually at the very end of the chain of reasoning--is always warned to the reader to leave
any minimally uncertain result open to debate. In section \ref{sec:discussion} the more mathematical
aspects of the reached solutions are left aside and the results are analyzed primarily from the
point of view of language dynamics: what does the analytical outcome mean in terms of coexisting
languages?

  \section{Derivation of the model equations} 
    \label{sec:derivation}
  
    Closely following the path pointed out by Abrams and Strogatz \cite{Abrams-Strogatz}, we work on
a model of two competing languages where bilingualism is an option in between and where the
similarity between the tongues plays and explicit role \cite{Mira-Paredes}. The model considers a
human population whose individuals might talk either of two languages $X$ or $Y$ or both of them.
Along the text we might refer to either of the monolingual communities or the bilingual one as
\emph{groups}, standing for \emph{groups of speakers}. Naming $x$ and $y$ the fraction of
monolingual speakers of each tongue, and naming $b$ the fraction of bilingual speakers; two non-
linear coupled differential equations are derived from the following basic hypothesis:
      \begin{enumerate}
        \item \emph{Population size remains constant.} 

        \item \emph{The probability that an individual acquires a language different from its
current one grows with the} status \emph{of the new language.} Therefore a status parameter
$s\in[0,1]$ of one of the languages is introduced (being $1-s$ the status of the other one). These
statuses are constant and a property of the system of coexisting languages.

        \item \emph{It is possible to define a distance between two languages}. This was done in
\cite{Mira-Paredes} introducing a parameter called interlinguistic similarity, $k\in[0,1]$: $k=0$
for \emph{orthogonal} languages, $k=1$ for exactly equal languages--a measure of how close the
languages are to each other. This allowed to get this distance in an easy and straight way, by
simply fitting percentages of speakers of the involved languages along time to a set of differential
equations. This interlinguistic similarity is constant and, again, a property of each pair of
languages. This parameter describes how difficult it is for a monolingual speaker to learn the other
language: this task should be easier if the languages are more similar to each other. Another view
of it is that the probability that an individual retains its old language when learning a new one
grows with this parameter $k$. Thus, the parameter $k$ is the gate which opens the way to the birth
of a bilingual group.
 
        \item \emph{The probability that an individual acquires a language different from its
current one grows with the fraction of speakers of the new language.} An exponent $a$ is introduced
to ponder the importance of the fraction of speakers in a group in attracting new speakers with
respect to $s$ and $k$. Once more, $a$ should be a constant that characterizes each system of two
competing languages. Existing field work with similar equations \cite{Abrams-Strogatz} shows that
the equivalent parameter in that model varies little across pairs of coexisting languages,
suggesting that social pressure to shift languages might be a constant through cultures.
      \end{enumerate}

    Considering different hypothesis or variations on the implementation of the current ones might
lead to different modeling of the same phenomenon \cite{Wang-Minett, Castello-SanMiguel, 
Minett-Wang}. In the present paper we focus on this minimal model whose results could be compared 
to data from a real system where the hypotheses are reasonably met \cite{Mira-Paredes, 
Mira-Nieto}. \\
   
    As stated previously, we seek to analytically solve the model introduced in \cite{Mira-Paredes}
and test the consistence of the results obtained in \cite{Mira-Nieto}. Therefore we will be working
on the set of differential equations that the authors derived for the dynamics of the fraction of
monolingual speakers of each language. The derivation is as follows: 

    We formulate the probability of shifting languages $P_{XY}$ and $P_{YX}$ and the probability of
arriving to (departing from) the bilingual group from each monolingual group $P_{XB}$, $P_{YB}$
($P_{BX}$, $P_{BY}$) based on the hypothesis listed above:
      \begin{eqnarray}
        P_{XB} &=& ck(1-s)(1-x)^a, \nonumber \\
        P_{YB} &=& cks(1-y)^a, \nonumber \\ 
        P_{BX} = P_{YX} &=& c(1-k)s(1-y)^a, \nonumber \\
        P_{BY} = P_{XY} &=& c(1-k)(1-s)(1-x)^a. 
        \label{eq:res.01}
      \end{eqnarray}
$c>0$ is a normalization constant \cite{Mira-Nieto}. Equations \ref{eq:res.01} are used to reckon
the rates at which the population of the three groups grow or decline:
      \begin{eqnarray}
        {dx\over dt} &=& F_x(x,y), \nonumber \\ 
        {dy\over dt} &=& F_y(x,y). 
        \label{eq:res.02}
      \end{eqnarray}
A third differential equation exists that tracks the evolution of the proportion $b$ of bilinguals,
but this will not be needed in the following thanks to the normalization of the population
$x+y+b=1$. In \cite{Mira-Nieto} these equations are written such that the contributions of all
terms from equations \ref{eq:res.01} can be explicitly read, but here we prefer a more compact
notation so that we can research the field $F=(F_x,F_y)$:
      \begin{eqnarray}
        F_x(x,y) &=& c\left[(1-x)(1-k)s(1-y)^a - x(1-s)(1-x)^a \right], \nonumber \\ 
        F_y(x,y) &=& c\left[(1-y)(1-k)(1-s)(1-x)^a - ys(1-y)^a \right]. 
        \label{eq:res.03}
      \end{eqnarray}
For the present study, the parameters are restricted to $k\in(0,1)$, $s\in(0,1)$, and $a>0$.
Sometimes $a$ will be further restricted. If this is the case, it will be noted.

    As said before, the parameter $c$ normalizes the dynamics and is irrelevant for equilibrium
points and stability issues, therefore it was not paid much attention in previous literature and
neither will it be paid attention now. The parameter $a$ generalizes the monotonously increasing
dependence of the probability of transition between languages as outlined in the fourth hypothesis.
This parameter has been found to be larger than $1$ and relatively constant among cultures
($a\sim1.31$) in experimental accounts of the problem of language dynamics \cite{Abrams-Strogatz,
Mira-Paredes, Mira-Nieto}, but we intend to address the behavior of the system for $a>0$, which is
a more interesting generalization. The other parameters that concern us are $k$ and $s$.

  \section{Resolution of the model} 
    \label{sec:resolution}
  
    \subsection{The model yields realistic trajectories}
      \label{sec:res.realTrajec}
   
      We note that all the possible distributions of speakers among the different groups can be
represented in the $x-y$ space. There, the condition $x+y+b=1$ defines a triangular set $A=\{(x,y),
\>\>x\ge0, y\ge0, x+y\le1\}$ upon which the fields $F_x(x,y)$ and $F_y(x,y)$ are acting. For the
sake of basic consistency of the model, its solutions must be feasible; meaning that a negative
number of individuals in any group should be forbidden: the dynamics must happen inside $A$ for
realistic systems.
        \begin{lemma}
          Assume the parameter $a>0$. The set $A=\{(x,y), \>\>x\ge0, y\ge0, x+y\le1\}$ is positive 
invariant. 
          \label{lem:01}
        \end{lemma}

        \begin{proof}
          Let us see that the field defined in equations \ref{eq:res.03} is directed inwards in the
          boundaries of $A$.
            \begin{enumerate}
              \item \textbf{If $x=0$ and $y\in[0,1]$}: 
                \begin{eqnarray}
                  F_x(0,y) &=& c(1-k)s(1-y)^a \ge 0, \nonumber\\ 
                  F_y(0,y) &=& c\left[ (1-y)(1-k)(1-s) - ys(1-y)^a \right]. 
                  \label{eq:res.04}
                \end{eqnarray}
              The first inequality implies that the field flows inwards $A$ since $F_x(0,y)\ge0$ 
              for $0\le y\le 1$. 

              \item \textbf{If $y=0$ and $x\in[0,1]$}: 
                \begin{eqnarray}
                  F_x(x,0) &=& c\left[ (1-x)(1-k)s - x(1-s)(1-x)^a \right], \nonumber\\ 
                  F_y(x,0) &=& c(1-k)(1-s)(1-x)^a\ge0. 
                  \label{eq:res.05}
                \end{eqnarray}
              For $y=0$ the field flows inwards $A$ since $F_y(x,0)\ge0$ for any $0\le x\le 1$. 

              \item \textbf{If $y=1-x$}: 
                \begin{eqnarray}
                  F_x(x,1-x) &=& c\left[ (1-x)(1-k)sx^a - x(1-s)(1-x)^a \right], \nonumber \\ 
                  F_y(x,1-x) &=& c\left[ x(1-k)(1-s)(1-x)^a - (1-x)s x^a\right]. 
                  \label{eq:res.06}
                \end{eqnarray}
              In this case the field flows inwards $A$ if and only if $-F_x(x,1-x)-F_y(x,1-x)\ge 0$
since $(-1,-1)$ is a normal vector to the straight line $y=1-x$ pointing towards the interior of
$A$. That condition is trivially satisfied and we conclude again that the field flows inwards into
$A$ through the segment $y=1-x$, $x\in(0,1)$.
            \end{enumerate}
        \end{proof}
    
      This lemma means that any real distribution of speakers between the available groups that
would evolve according to the proposed equations would remain feasible all the time.
    
    \subsection{Number of fixed points of the dynamics in $A$} 
      \label{sec:res.numberFixed}

      It is possible to find upper and lower limits to the number of equilibrium points that the
system displays for different $k$, $s$, and $a$. We find the equilibrium point of the system
wherever the nullclines (curves defined by $dx/dt=0$ and $dy/dt=0$) intersect each other. In $A$,
the equilibrium points that can be detected by a simple inspection of the system are $P_x=(1,0)$ and
$P_y=(0,1)$. We will term them \emph{trivial} fixed points. There is another trivial fixed point for
the dynamics, but it lays outside $A$: the point $(1,1)$. Depending on different values of the
mentioned parameters we shall find more equilibrium points inside $A$. Following the notation
introduced in this paragraph, we appreciate that the curves $x=1$ and $y=1$ are branches of the
nullclines of the system. We name them the \emph{trivial} branches. Our analysis will deal mainly
with the non trivial branches.
   
      With a preliminary analysis of the nullclines of the field ($F_x=0$, $F_y=0$) it is possible
to narrow down the number of fixed points in the interior of $A$ to a maximum of $3$: Equilibrium
points of the dynamics are found in the intersections of the nullclines. Equating both components
of the field to zero we get:
        \begin{eqnarray}
          (1-y)^a &=& {1-s\over s}{1\over 1-k}x(1-x)^{a-1}, \nonumber \\
          (1-x)^a &=& {s\over 1-s}{1\over 1-k}y(1-y)^{a-1}. 
          \label{eq:res.07}
        \end{eqnarray}
Multiplying these equations and isolating $y$: 
        \begin{eqnarray}
          y = {1\over 1+{x\over (1-k)^2(1-x)}}. 
          \label{eq:res.08}
        \end{eqnarray}
We must restrict ourselves to $x\ne 1 \ne y$ now to avoid divergences here and in following
equations, but this is enought to continue with our discussion.

      Substituting equation \ref{eq:res.08} into the second expression of equations \ref{eq:res.07}: 
        \begin{eqnarray} 
          (1-x)^a &=& {s(1-k)(1-x)\over (1-s)x}\left(1-{1\over 
          1+{x\over(1-k)^2(1-x)}} \right)^a \Rightarrow \nonumber \\
          &\Rightarrow& \left( 1-x\over x \right)^{a-1}\left(x+(1-k)^2(1-x)\right)^a = {s\over 1-s}(1-k). 
          \label{eq:res.09}
        \end{eqnarray}
If $(x^*,y^*)$ is an equilibrium point, $x^*$ must obey equation \ref{eq:res.09} and the
corresponding $y^*$ is obtained from equation \ref{eq:res.08}.

      From the left-hand side of equation \ref{eq:res.09}: 
        \begin{eqnarray} 
          g(x) &\equiv& \left(1-x\over x\right)^{a-1}\left(x+(1-k)^2(1-x)\right)^a. 
          \label{eq:res.10}
        \end{eqnarray}
If $a>1$ it is true that: 
        \begin{eqnarray}
          g(x) &>& 0, \>\>\>\> \forall x\in(0,1); \nonumber \\ 
          \lim_{x \to 0^+} g(x) &=& +\infty; \nonumber \\ 
          \lim_{x \to 1^-} g(x) &=& 0. 
          \label{eq:res.11}
        \end{eqnarray}
Furthermore, $g(x)$ has got a relative minimum and a relative maximum respectively at: 
        \begin{eqnarray}
          x^- &=& {1\over 2a}\left( 1-\sqrt{1-{4a(a-1)(1-k)^2\over 2k-k^2}} \right)\in(0,1), \nonumber \\ 
          x^+ &=& {1\over 2a}\left( 1+\sqrt{1-{4a(a-1)(1-k)^2\over 2k-k^2}} \right)\in(0,1). 
          \label{eq:res.12}
        \end{eqnarray}
Because all of this, the equation $g(x) = {s\over 1-s}(1-k)$ can only have one, two, or three
solutions in $x\in(0,1)$ for fixed $k$ and $s$, restricting thus the number of equilibrium points of
the whole system. 

      For $a=1$, $g(x)$ reduces to a straight line that might or might not fulfill $g(x)={s\over
1-s}(1-k)$ within the range of interest $x\in(0,1)$. Because $g(x)$ is a straight line, this
equality can be obeyed for just one value of $x$ at most. Thus for $a=1$ there is at most one more
fixed point within $A$, but it must not necessarily exist. 

      Finally, for $a<1$ the limits found in equation \ref{eq:res.11} swap: 
        \begin{eqnarray}
          \lim_{x \to 0^+} g(x) &=& 0, \nonumber \\ 
          \lim_{x \to 1^-} g(x) &=& +\infty; 
          \label{eq:res.13}
        \end{eqnarray}
and $g(x)>0$ in the whole range $x\in(0,1)$. It is also monotonically increasing within this range
and thus must always match ${s\over 1-s}(1-k)$ in exactly one point $x^*$ internal to $A$. So for
$a<1$ there is always one fixed point besides the trivial ones.

    \subsection{Stability of the equilibrium points $P_x=(0,1)$ and $P_y=(1,0)$ for $a>1$} 
      \label{sec:res.stab01}
   
      We assess the stability of the system by evaluating the matrix at the existing fixed points,
diagonalizing it, and considering the sign of the eigenvalues. For non trivial equilibrium points it
becomes complicated to exactly locate them on the $x-y$ plane, left asside its analysis through the
Hessian matrix; but for $P_x$ and $P_y$ and restricting ourselves to $a>1$ we can evaluate the
Hessian matrix explicitly and it happens to be diagonal already: 
        \begin{eqnarray}
          DF(0,1) &=& \left( \begin{array}{cc}
                -c(1-s) & 0 \\ 
                0 & -c(1-k)(1-s)
               \end{array} \right), \nonumber \\ 
          DF(1,0) &=& \left( \begin{array}{cc}
                -c(1-k)s & 0 \\ 
                0 & -cs
               \end{array} \right). 
          \label{eq:res.14}
        \end{eqnarray}
Furthermore, the eigenvalues are negative meaning that $P_x$ and $P_y$ are asymptotically stable
independently of the values of $k$ and $s$ for $a>1$.

      Because $P_x$ and $P_y$ are always stable for $a>1$ and the field is such that all
trajectories enter $A$, \emph{if} there is \emph{only one} more equilibrium point $x^*$ interior to
$A$ it must be a saddle point and lie exactly at the frontier between the basins of attraction of
$(0,1)$ and $(1,0)$. If $x^*$ were unstable yet not a saddle point, either there would exist two
more fixed points where the boundaries between basins cross the frontier of $A$, or there would
exist trajectories leaving $A$; and neither of these is the case. If $x^*$ were stable it would have
a basin of attraction for itself and new fixed points would need to exist in the separation between
different basins.

    \subsection{Studying the field in different regions of $A$ for $a>1$}
      \label{sec:res.field}
   
      Now we will get more insights about the dynamics by further characterizing the field $F$ at
the boundary of $A$ and in its interior. For this analysis we must assume $a>1$, otherwise some of
the functions that we will be making use of will be ill-defined.

      \subsubsection{Studying $F$ at the boundary of $A$}
        \label{sec:res.field.boundary}
	
        Recalling $F_y(0,y)$ from equation \ref{eq:res.04}, we introduce: 
          \begin{eqnarray}
            G_y(y) &\equiv& (1-k)(1-s) - ys(1-y)^{a-1}, 
            \label{eq:res.15}
          \end{eqnarray}
and we note that it is continuous on the interval $[0,1)$, strictly decreasing on $(0,1/a)$ and
strictly  increasing on $(1/a,1)$. Since:
          \begin{eqnarray}
            G_y(0) = G_y(1) &=& (1-k)(1-s)>0, 
            \label{eq:res.16}
          \end{eqnarray}
we can find out if this function ever changes its sign by evaluating it at its minimum: $G_y(1/a)$.
We get either:
          \begin{eqnarray}
            {(a-1)^{a-1} \over a^a} &<& (1-k){1-s\over s}, 
            \label{eq:res.17}
          \end{eqnarray}
which would imply that $F_y(0,y)>0 \>\> \forall y\in[0,1)$; or: 
          \begin{eqnarray}
            {(a-1)^{a-1} \over a^a} &\ge& (1-k){1-s\over s}, 
            \label{eq:res.18}
          \end{eqnarray}
which would imply that there would exist $y_1, y_2\in(0,1)$ such that $F_y(0,y)>0$ if
$y\in[0,y_1)\cup(y_2,1]$ and $F_y(0,y)<0$ for all $y\in(y_1,y_2)$. In this case $F_y$ is zero at
$(0, y_1)$ and $(0,y_2)$. Let us note that $y_1=y_2$ if the equality holds on equation
\ref{eq:res.18}. 

        Likewise, recalling $F_x(x,0)$ from equation \ref{eq:res.05}, we define: 
          \begin{eqnarray}
            G_x(x) &\equiv& (1-k)s-x(1-s)(1-x)^{a-1}, 
            \label{eq:res.19}
          \end{eqnarray}
which is continuous on $x\in[0,1)$, strictly decreasing on $(0,1/a)$ and strictly increasing on
$(1/a,1)$.  Also:
          \begin{eqnarray}
            G_x(0) = G_x(1) &=& (1-k)s>0, 
            \label{eq:res.20}
          \end{eqnarray}
thus we find either: 
          \begin{eqnarray}
            {(a-1)^{a-1} \over a^a} &<& {s\over 1-s}(1-k), 
            \label{eq:res.21}
          \end{eqnarray}
which would imply $F_x(x,0)>0 \>\> \forall x\in[0,1)$; or: 
          \begin{eqnarray}
            {(a-1)^{a-1} \over a^a} &\ge& {s\over 1-s}(1-k), 
            \label{eq:res.22}
          \end{eqnarray}
which would imply that there would exist $x_1, x_2\in(0,1)$ such that $F_x(x,0)>0$ if
$x\in[0,x_1)\cup(x_2,1]$ and $F_x(x,0)<0$ for all $x\in(x_1,x_2)$. In this case $F_x$ is zero at
$(x_1,0)$ and $(x_2,0)$. Once again: $x_1=x_2$ if the equality holds on equation \ref{eq:res.22}.

        Let us note that if the strict inequalities \ref{eq:res.17} and \ref{eq:res.22} are
simultaneously true then $s<1/2$ and if the strict inequalities \ref{eq:res.18} and \ref{eq:res.21}
are simultaneously true then $s>1/2$. \\

        On the diagonal $x+y=1$ the field takes the form written in equation \ref{eq:res.06} and the
signs of $F_x(x,1-x)$ and $F_y(x,1-x)$ can be studied analyzing $\left({1-x\over x}\right)^{a-1}$
and $\left({x\over 1-x}\right)^{a-1}$ respectively: 
   
        Since the function $\left({1-x\over x}\right)^{a-1}$ is strictly decreasing on $(0,1)$ and: 
          \begin{eqnarray}
            \lim_{x \to 1^-} \left({1-x\over x}\right)^{a-1} &=& 0, \nonumber \\ 
            \lim_{x \to 0^+} \left({1-x\over x}\right)^{a-1} &=& +\infty; 
            \label{eq:res.23}
          \end{eqnarray}
then there exists only one $z_x\in(0,1)$ such that $\left({1-z_x\over z_x}\right)^{a-1}=(1-k){s\over
1-s}$ for fixed $k$ and $s$. Therefore $F_x(x,1-x)<0$ if $x\in(0,z_x)$ and $F_x(x,1-x)>0$ if
$x\in(z_x,1)$.

        With a similar argument for $\left({x\over 1-x}\right)^{a-1}$ it can be warranted the
existence of only one $z_y\in(0,1)$ such that $\left({z_y\over 1-z_y}\right)^{a-1} = (1-k){1-s\over
s}$ for fixed $k$ and $s$. Then $F_y(x,1-x)>0$ if $x\in(0,z_y)$ and $F_y(x,1-x)<0$ if $x\in
(z_y,1)$. 

        It can be trivially shown that $z_y<z_x$. Also it is true that
$\left|F_x(x,1-x)\right|>F_y(x,1-x)$ if $x\in (0,z_y)$ and $F_x(x,1-x) < \left| F_y(x,1-x) \right|$
if $x\in (z_x,1)$.

      \subsubsection{Further study of $F_x$} 
        \label{sec:res.field.Fx}
	
        If $x\ne0$ and $x\ne1$, then the points $(x,y)$ which nullify the first component of the
field $F$ are those that obey:
          \begin{eqnarray}
            y &=& 1-\left({1-s\over s}\right)^{1/a} {1\over (1-k)^{1/a}}x\left({1-x\over x}\right)^{1-1/a}. 
            \label{eq:res.24}
          \end{eqnarray}
The function $h_x:[0,1]\rightarrow \mathbb{R}$ is defined as: 
          \begin{eqnarray}
            h_x(x) &=& 1-\left({1-s\over s}\right)^{1/a} {1\over (1-k)^{1/a}}x\left({1-x\over x}\right)^{1-1/a} 
            \label{eq:res.25}
          \end{eqnarray}
on $x\in(0,1]$ and $h_x(0)=1$. It is strictly decreasing on $(0,1/a)$ and increasing on $(1/a,1)$,
it has got a minimum at $1/a$, and $h_x(0) = 1 = h_x(1)$.

        If $y<h_x(x)$ then $F_x(x,y)>0$. If $y=h_x(x)$ then $F_x(x,y)=0$. If $y>h_x(x)$ then
        $F_x(x,y)<0$.

        If the parameters of the system are such that inequality \ref{eq:res.21} holds, then
$h_x(1/a)>0$ and the plot of $h_x(x)$ intersects the boundary of $A$ at $(0,1)$ and $(z_x, 1-z_x)$
(figs. \ref{fig:aFixed}\textbf{a}-\textbf{b}). 

        If the strict inequality \ref{eq:res.22} holds true then $h(1/a)<0$ and the plot of $h_x(x)$
intersects the boundary of $A$ at $(0,1)$, $(x_1,0)$, $(x_2,0)$, and $(z_x,1-z_x)$ (figs.
\ref{fig:aFixed}\textbf{c}-\textbf{e}).

        Additionally, since $F_y(x_2, h_x(x_2))>0$ and $F_y(z_x,h_x(z_x))<0$, because $F_y$ and
$h_x$ are continuous, it is warranted the existence of $p\in(x_2,z_x)$ such that $F(p,h_x(p))=(0,0)$
--i.e. a fixed point of $F$. 

      \subsubsection{Further study of $F_y$} 
        \label{sec:res.field.Fy}
	
        If $y\ne0$ and $y\ne1$, then the points $(x,y)$ which nullify the second component of the
field $F$ are those which obey:
          \begin{eqnarray}
            x &=& 1-\left( s\over1-s \right)^{1/a}{1\over (1-k)^{1/a}}y^{1/a}(1-y)^{1-1/a}. 
            \label{eq:res.26}
          \end{eqnarray}
We define $h_y(y)$ similarly as we defined $h_x(x)$. This function is strictly decreasing on
$(0,1/a)$ and increasing on $(1/a,1)$, and $h_y(0)=1=h_y(1)$.

        If inequality \ref{eq:res.17} holds true then $h_y(1/a)>0$ and the curve $\{(h_y(y),y),
y\in[0,1]\}$ intersects the boundary of $A$ at $(1,0)$ and $(z_y,1-z_y)$ (fig.
\ref{fig:aFixed}\textbf{b}).

        If the strict inequality \ref{eq:res.18} holds true then $h_y(1/a)<0$ and the curve
$\{(h_y(y),y), y\in[0,1]\}$ intersects the boundary of $A$ at $(1,0)$, $(0,y_1)$, $(0,y_2)$, and
$(z_y,1-z_y)$ (figs. \ref{fig:aFixed}\textbf{a}, and \ref{fig:aFixed}\textbf{c}-\textbf{e}).

        Additionally, since $F_x(h_y(y_2), y_2)>0$ and $F_x(h_y(1-z_y),1-z_y)<0$, because $F_x$ and
$h_y$  are continuous, it is warranted the existence of $q\in(y_2,1-z_y)$ such that
$F(h_y(q),q)=(0,0)$ --i.e. a fixed point of $F$. \\

        The evolution of both $h_x(x)$ and $h_y(y)$ as a function of the parameters $a$, $s$, and
$k$ is partially shown in figs. \ref{fig:aFixed} and \ref{fig:aVary}, and can be dynamically
explored in \cite{Maquez}.

    \subsection{The nature of the orbits help us assess the stability of non-trivial fixed points} 
      \label{sec:res.field.orbit}
   
      The nullclines are always landmarks of the dynamic system under research. Their obvious use is
to locate the equilibrium points in their intersections, but more information can be extracted if we
look at them carefully. In section \ref{sec:res.field} we used them to find out how the field
behaves in the boundaries of $A$ as they mark the sets of points where the vertical and horizontal
components of the field are nullified. This applies also in the interior of $A$: The trajectories of
the system pass by with vertical tangent through the points of the curve $F_x(x,y)=0$, and with
horizontal tangent through the points of the curve $F_y(x,y)=0$. But also, these curves divide $A$
in regions within which the signs of $F_x$ and $F_y$ are well determined. Topological arguments
regarding the action of $F$ upon these different regions of $A$ help us put some limits to the kind
of orbits that the system can yield: we will see that periodic dynamics can be banned. These
considerations also let us find out whether non-trivial points are stable or not for $a>1$. 

      In this range of $a$ the system will always have at least one more equilibrium point in the
interior of $A$. We have seen that this can be deduced either from the crossings of $h_x(x)$ and
$h_y(y)$ with the boundary of $A$ or from equation \ref{eq:res.09} attending to the shape of $g(x)$.
The analysis of $g(x)$ let us further know that also two or at maximum three fixed points can exist
inside $A$. These three, two, or one equilibrium points will show up depending on the values of the
parameters $s$, $k$, and $a$. Many possibilities are illustrated in fig. \ref{fig:aFixed} and in
\cite{Maquez}.

      With this in mind, let us consider the following regions: 
        \begin{eqnarray}
          R_1 &=& \{ (x,y)\in A / F_x(x,y)<0, F_y(x,y)>0 \}, \nonumber \\ 
          R_2 &=& \{ (x,y)\in A / F_x(x,y)>0, F_y(x,y)<0 \}. 
          \label{eq:res.27}
        \end{eqnarray}
Equivalently: 
        \begin{eqnarray}
          R_1 &=& \{ (x,y)\in A / x\le h_y(x), y\ge h_x(x) \}, \nonumber \\ 
          R_2 &=& \{ (x,y)\in A / x\ge h_y(x), y\le h_x(x) \}. 
          \label{eq:res.28}
        \end{eqnarray}
In a similar way we could introduce: 
        \begin{eqnarray}
          B_1 &=& \{ (x,y)\in A / F_x(x,y)>0, F_y(x,y)>0 \}, \nonumber \\ 
          B_2 &=& \{ (x,y)\in A / F_x(x,y)<0, F_y(x,y)<0 \}; 
          \label{eq:res.29}
        \end{eqnarray}
but these will not be interesting for us right now. 

      Focusing on $R_1$ and $R_2$, they have got one or two connected components depending on if
inside $A$ there exist one, two, or three equilibrium points. We shall write $R_1=A_1\cup A_3$,
$R_2=A_2\cup A_4$; being $A_3$ or $A_4$ empty if on the interior of $A$ there are not three
equilibrium points, and $A_1$ and $A_2$ the regions whose boundaries contain respectively
$P_y=(0,1)$ and $P_x=(1,0)$. An account of these regions for some values of the parameters can be
seen in fig. \ref{fig:regions}.

      Taking into account the sign of the components of the field we can tell that regions $A_1$,
$A_2$, $A_3$, and $A_4$ are positive invariant. If $(x_0,y_0)\in A_1$ then its trajectory
$(x(t),y(t))$ for $t\in I_t[0,\infty)$ lays in $A_1$ because in the boundary of $A_1$ the field
points inwards. This trajectory is thus contained in a compact for $t\in I_t$. Also, since
${dx(t)\over dt}<0$ and ${dy(t)\over dt} >0$, it can be verified that $x(t)$ is monotonously
decreasing on $I_t$ and $y(t)$ is monotonously increasing on $I_t$. Consequently it exists the limit
$\lim_{t \to +\infty} (x(t),y(t))=P_y$. The set $A_1$ is therefore contained in the basin of
attraction of $P_y$. Analogously, it can be shown that if $(x_0,y_0)\in A_2$ then $\lim_{t \to
+\infty} (x(t),y(t))=P_x$: its trajectory lays in the region of attraction of $P_x$; thus, the basin
of attraction of $P_x$ contains $A_2$. 

      Also, if $A_3$ and $A_4$ are both non-empty and $(x_0,y_0)\in A_3\cup A_4$ its trajectory
remains either inside $A_3$ or inside $A_4$ and converges towards the equilibrium point at the
intersection of the frontiers of these regions. Since $A_3$ and $A_4$ are non-empty only when there
are three equilibrium points inside $A$, this result means that one of these three points, whenever
they exist, must be stable. In this case we can determine that the two remaining fixed points in the
interior of $A$ must be saddle points. We do so with an argument similar to the one we used to show
that $x^*$ is a saddle point when only one equilibrium point exists inside $A$ (section
\ref{sec:res.stab01}). We further deduce that the case with two interior fixed points corresponds to
a saddle-node bifurcation and that this situation is the frontier between those cases with one and
three equilibrium points in the space of parameters $k-s$. 

      Concerning the dynamics of the system, it is important the following lemma which ensures that
there is no oscillatory behavior:
        \begin{lemma}
          There are not any periodic orbits in the x-y plane. 
          \label{lem:02}
        \end{lemma}
      	\begin{proof}
          Actually, because of the regularity of the field, applying the Poincar\'e-Bendixson
theorem it can be deduced that if there would exist any closed orbit it must enclose a fixed point
on its interior. Thus, the periodic orbit would necessarily enter and exit two of the regions $A_i$.
This cannot happen because all regions $A_i$ are positive invariant. 
        \end{proof}

      This same argument also implies that fixed points cannot be foci, because trajectories
approaching them  should cross many times the frontiers between regions, some of which are positive
invariant and cannot be left.

    \subsection{Tentative solutions for $a\le 1$}
      \label{sec:sola1}
   
      Splitting the problem in $a>1$ on the one side and $a\le 1$ on the other made its solution
easier because several of the reasonings that work very well in the former case are built on
functions that are ill-defined in the later. An example are the functions $G_x(x)$ and $G_y(y)$, but
also the Hessian matrix in $(1,0)$ and $(0,1)$ present some problems for $a<1$. Luckily enough, in
section \ref{sec:res.numberFixed} we proved that there is just one more equilibrium point
$(x^*,y^*)$ for $a\le 1$ which always appears if $a<1$ and that might not appear for $a=1$ depending
on the parameters, so we do not need to investigate $3$ prospective fixed points as for $a>1$. Also
it is still valid the demonstration that $A$ is positive invariant made in section
\ref{sec:res.realTrajec}. 

      In fig. \ref{fig:aVary} the nullclines are represented for various values of $a$ and fixed $k$
and$s$. We can observe $(x^*,y^*)$ in the intersections, and we can also observe how the nullclines
suffer a deep transformation as values of $a$ larger than $1$ are employed. We shall study now the
cases $a=1$ and $a<1$. Because of the analytic results are not so satisfactory, we shall complement
them using numerical simulations whenever it is useful. These results should be questioned as long
as a complete mathematical proof is not available. \\

      It is particularly illustrative the resolution of the stability of $P_x=(1,0)$ and $P_y=(0,1)$
for $a=1$. This can still be analytically done. The Hessian matrix in this case reads:
        \begin{eqnarray}
          DF(0,1) &=& \left( \begin{array}{cc}
                -c(1-s) & -c(1-k)s \\ 
                0 & -c(1-k)(1-s)+cs
               \end{array} \right), \nonumber \\ 
          DF(1,0) &=& \left( \begin{array}{cc}
                -c(1-k)s + c(1-s) & 0 \\ 
                -c(1-k)(1-s) & -cs
               \end{array} \right). 
          \label{eq:res.30}
        \end{eqnarray}
The eigenvalues are $\lambda_{P_y}^1=-c(1-s)$ and $\lambda_{P_y}^2=-c(1-k)(1-s)+cs$ for $P_y$ and
$\lambda_{P_x}^1= -c(1-k)s + c(1-s)$ and $\lambda_{P_x}^2=-cs$ for $P_x$. We see that one of the
eigenvalues is always the sum of two terms with different sign and this compromises the stability of
$P_x$ and $P_y$. Indeed, their stability depends now on the parameters $k$ and $s$. By equating the
conflictive terms to zero we obtain two curves relating $k$ and $s$: 
        \begin{eqnarray}
          s_{P_y} &=& {1-k_{P_y}\over 2-k_{P_y}}, \nonumber\\
          s_{P_x} &=& {1\over 2-k_{P_x}}. 
          \label{eq:res.31}
        \end{eqnarray}
These curves tell us where does the stability of $P_x$ and $P_y$ change in the space of parameters
$k-s$: $P_x$ is stable for $s>s_{P_x}$ and $P_y$ is stable for $s<s_{P_y}$. There is a region of
values $s_{P_y} < s < s_{P_x}$ where neither $P_x$ nor $P_y$ are stable and, since there are not
any trajectories leaving $A$, there must exist a point $(x^*,y^*)$ interior to $A$ that is stable.

      The curves $s_{P_x}(k_{P_x})$ and $s_{P_y}(k_{P_y})$ are plotted in fig.
\ref{fig:phases}\textbf{a}. There it is shown the stability of the different stable points for
different values of $a$ in the $k-s$ space, but aided by computer simulations. We see that the
numerical results match the analytical results for $s_{P_x}(k_{P_x})$ and $s_{P_y}(k_{P_y})$, and
that these curves seem to evolve into the boundaries between different regimes as $a$ takes values
larger than $1$. \\

      For $a<1$ it is not possible to work out the stability of $P_x$ nor $P_y$ in a rigorous way.
The Hessian matrix has diverging terms in this case: it is not well defined. We know, though, that
there is always a third fixed point inside $A$ and computer simulations suggest that this interior
equilibrium point is always stable for $a<1$. This would be in agreement with results obtained for
the simpler precursor model \cite{Abrams-Strogatz} where a similar stable point is found for $a<1$
\cite{Vazquez-SanMiguel, Patriarca-SanMiguel}.

    \section{Discussion and conclusions} 
      \label{sec:discussion}
  
      The analytical results that we reached here are in partial discordance with those reported
from previous numerical works \cite{Mira-Nieto}. Although the number of fixed points reckoned now
could agree with previous accounts, the stability of them has been misread because of reasons that
will be obvious right now. Notwithstanding this, the conclusions from \cite{Mira-Nieto} remain
largely the same, as we will see. We first discuss the results for $a>1$ to compare directly with
previous works and then we make some remarks about the cases $a\le1$. 

      In \cite{Mira-Nieto} the stability of the system was assessed through computer simulations
only, thus stable fixed points were partly identified. Non-stable or saddle points did not stand out
in these simulations because specific tests were not run therefore. It was concluded that one, two,
or three stable equilibrium points existed in $A$, including the trivial ones at its boundary:
$P_x$, $P_y$, and $P^*=(x^*,y^*)$; the later being the only equilibrium point explicitly detected in
the interior of $A$. This matches exactly the picture drawn from the current work. The discrepancies
arise regarding the stability of $P_x$ and $P_y$: The computational tools used in \cite{Mira-Nieto}
yielded a result that strongly suggested that this stability depended upon $k$ and $s$ (dependence
upon $a$ was not addressed: it was taken $a=1.31$ for historical reasons) and that $P^*$ was stable
whenever it could be detected by simulations--as it could only be detected if it was an attractor of
the dynamics in the discretized version of equations \ref{eq:res.02}.

      A plot was elaborated in \cite{Mira-Nieto} that divided the $k-s$ space of parameters in five
regions that would correspond to five different stability/instability combinations of the
equilibrium points. Namely: i) only $P_x$ or only $P_y$ is stable, $P^*$ is not detected; ii) both
$P_x$ and $P_y$ are stable, $P^*$ is not detected; iii) the three possible points are detected and
stable; iv) $P^*$ and either only $P_x$ or only $P_y$ are stable; and v) only $P^*$ is stable. In
all five cases $P_x$ and $P_y$ were supposed to exist and to be instable whenever their basins of
attraction were found empty by the computer simulations. An updated version of that plot is
reproduced in fig. \ref{fig:phases}\textbf{c} with the same five regions colored with different
shades of blue and green.

      The interpretation given in \cite{Mira-Nieto} was not right although it was consistent with
the numerical outcome. For example: we now know that $P_x$ and $P_y$ are always stable for $a>1$,
disregarding the values of $k$ and $s$. We also know that at least one equilibrium point exists
always in the interior of $A$, which may not be stable and which may lay in the boundary between the
basins of attraction of $P_x$ and $P_y$. In fig. \ref{fig:aFixed} they are shown the plots of the
nullclines for the exact same parameters as those used in \cite{Mira-Nieto} and we readily see how
for certain parameters some saddle points interior to $A$ approach $P_x$ or $P_y$ leading to a
reduction of their basins of attraction. We now know that these equilibrium points never collapse
into an unstable point. A basin of attraction may become undetectable to numerical means, and thus
$P_x$ and $P_y$ may be deemed instable; but we have now found out analytically that $P_x$ and $P_y$
remain stable for any value of $k$ and $s$ if $a>1$. We were able to reproduce this numerical effect
for the parameters used in \cite{Mira-Nieto} and for many others, as it can be seen in fig.
\ref{fig:phases}\textbf{b}-\textbf{d}: lighter shades of blue or green indicate sets of parameters
for which the computer simulations led to a wrong interpretation of the stability of some of the
fixed points.

      For $a>1$, the updated, more correct picture is as follows: There are always 2 stable trivial
fixed points $P_x$ and $P_y$ and depending on the parameters $k$ and $s$ there are $1$, $2$, or $3$
more fixed points in the interior of $A$. It can be shown that if three points exist, one of them
(termed $P^*$) must be stable. It can be argued that if there is only one fixed point inside $A$ it
must be a saddle point; and that if there are three, those equilibrium points different from $P_x$,
$P_y$, and $P^*$ must be saddle points as well. It can be guessed that the situation with $2$
equilibrium points inside $A$ corresponds to a saddle-node bifurcation as we transit through the
$k-s$ space from a region with one to a region with three non-trivial fixed points. The existing
analytical evidence and fig. \ref{fig:phases}, that shows the results of refined numerical
simulations, are consistent with this view; although those facts that were not analytically proven
in section \ref{sec:resolution} must be taken with enough care. 

      These equilibrium fixed points and their stability have got a direct interpretation for the
phenomenon for which the model was developed in the first place: they determine whether two
coexisting languages would remain alive together, or if one of them is going to take over and
extinguish the other. Also, the nature of the fixed points reached by the dynamics determines
whether individual bilingualism can be a stable trait. Given a pair of languages $X$ and $Y$ that
coexist with status $s_X=s$ and $s_Y=1-s$ and interlinguistic similarity $k$ in a society with a
fixed value of $a>1$, the model presents two well differentiated regions: 
        \begin{enumerate}

          \item \textbf{Coexistence is unstable}: one language ends up suppressing the other and the
bilingual group. What language survives depends on the initial distribution of speakers among
monolinguals of each language and bilinguals. This case corresponds to only one equilibrium point--
which turns out not to be stable--in the interior of $A$ and is depicted in figs.
\ref{fig:aFixed}\textbf{a}-\textbf{b}. The regions in the $k-s$ space where this happens are colored
in blue in fig. \ref{fig:phases}\textbf{b}-\textbf{d}.

          \item \textbf{Coexistence is possible} depending on the parameters $k$, $s$, and $a$; and
on the initial conditions of the dynamics. This case is the one with three stable equilibrium points
$P_x$, $P_y$, and $P^*$. The initial conditions determine whether a language drives the other to
extinction and makes bilingualism disappear, or if a steady state is reached ($P^*$) in which groups
of monolingual speakers of both languages survive along with a bilingual group. This is what happens
in figs. \ref{fig:aFixed}\textbf{c}-\textbf{e}; and parameters $k-s$ for which we find this
situation are indicated in green in figs. \ref{fig:phases}\textbf{b}-\textbf{d}.

        \end{enumerate}
A third case regarding number of fixed points would exist at the boundary between these regions, but
it does not seem to introduce any new behavior attending to the coexistence of languages. Coming
back to the incomplete interpretation made of the results in \cite{Mira-Nieto}, the two cases just
outlined already include those configurations of parameters for which some attractors are so small
that the extinction of a language or the coexistence of both of them is almost unavoidable without
regard of the initial conditions; although we now know that this is never the case. 

      We can see from the numerical simulations in fig. \ref{fig:phases} that regions of the $k-s$
space where stable bilingualism is possible correspond to those with a more balanced status between
languages. This balance is not so important for larger interlinguistic similarity: then a stable
bilingual situation can be reached even for well distinct $s_X$ and $s_Y$, depending on the initial
distribution of speakers. Further illustration of the role of $k$ and $s$ is made in fig.
\ref{fig:bifurcation}, that shows qualitative bifurcation diagrams of the stable fixed points when
varying these parameters with fixed $a$.

      The parameter $a$ was found relatively constant among cultures as indicated before \cite
{Abrams-Strogatz}, and this justified why it was not payed that much attention. But now we have also
studied how the possible outcomes change as $a$ varies. First, considering only $a>1$, we see (fig.
\ref{fig:phases}) that a larger $a$ means that the possibilities for stable bilingualism are
reduced. This parameter was already considered in an analysis of the more basic Abrams-Strogatz
model \cite{Chapel-SanMiguel, Vazquez-SanMiguel, Patriarca-SanMiguel}, and it was cleverly termed
\emph{volatility parameter}: the lower $a$ the more \emph{volatile} a large group becomes and vice-
versa. Thus, for larger $a$ bigger groups are more persistent and it is smaller the set of
parameters for which it can be reached a more diluted distribution of speakers (this would be: a
solution with speakers belonging to the bilingual groups, or communities with monolingual groups of
each language coexisting together). For lower $a$, larger monolingual groups are not so permanent
and a steady solution is easier to reach in which all languages coexist.

      This volatility is a critical feature at $a=1$: $P_x$ and $P_y$ become unstable for some of
the parameters $k$ and $s$  as it was said before. In terms of language dynamics this means that a
monolingual group agglutinating all the speakers is  no longer possible, whatever the initial
conditions, if their statuses are close enough and depending on the similarity  between languages.
The region of the $k-s$ space where this happens can be seen in green in fig.
\ref{fig:phases}\textbf{a}.  In such cases the only stable solution in the long term is the
coexistence of the monolingual groups along with the bilingual  one. But still at $a=1$ one language
might extinguish the other if its status is larger enough. There is a crucial difference  between
such an extinction and those happening for $a>1$: before, both languages could survive depending on
the initial number  of speakers of each one; now the extinction does not depend on this initial
condition if the parameters are those needed for  a language to take over (blue regions in fig.
\ref{fig:phases}\textbf{a}).

      The possibility of language extinction seems to change completely for $a<1$: then the
volatility is so high  that the monolingual options are never stable and the survival of both
languages within their monolingual groups  and along a bilingual group of speakers is guaranteed for
any values of the parameters $k$ and $s$, and for any  initial distribution of speakers. It was not
possible to prove this very last result analytically beyond any doubt  and it was obtained thanks to
numerical simulations. This solution is consistent with similar outcomes for the seminal  Abrams-
Strogatz model \cite{Chapel-SanMiguel, Vazquez-SanMiguel, Patriarca-SanMiguel}, which should be the
limit  case for $k\rightarrow0$ and $b\rightarrow0$ of the equations under research in this paper.

      \subsection{Nature of the orbits and higher order contributions} 

        A very important contribution of this paper is that brought in by lemma \ref{lem:02}. It is
clear its mathematical meaning: because of the nature of the field $(F_x,F_y)$ around a fixed point
it is not possible to find closed--i.e. periodic--orbits. The same lemma also implies that any
solution must consist of an exponential decay towards a fixed point: that an oscillatory decay is
not possible. The interpretation of this result is rather strong when it comes down to languages:
the extinction or raise of languages must be a monotonous phenomenon according to the present
equations. Tendencies that could be expected, e.g. alternation in the preponderance of a language in
a region, should not be observed. If such result were derived from real data, research should be
focused on what is needed to complement the present model: ``what would be the minimum elements that
play a role in generating cyclic behavior in language competition?'', because those employed here
would not suffice. 

        This lemma has got also some implications even if we would consider higher order or
stochastic extensions of the present model. The equations investigated in this paper are nothing but
a deterministic, mean field approach to  a phenomenon that usually takes place on a stochastic
environment. The next more realistic strategy to  model language competition or coexistence
departing from our current equations would be simulations of  discrete agents that shift between the
monolingual or the bilingual groups at random, being the transition  probabilities given by
equations \ref{eq:res.01}. This is coherent if we suppose some free-will and  variability to the
speakers when deciding what language to use. We expect thus intrinsic stochasticity  to be present
and manifest throughout noise. The power spectrum of this noise can be investigated. The  nature of
the equilibrium points found for equations \ref{eq:res.02} establishes some important limitations
to the kind of dynamics that can arise, as we will argue.

        Several techniques are available in the literature to incorporate uncertainty in a
deterministic  model, from agent-based simulations with stochastic interaction events
\cite{Gillespie1976, Gillespie2007}  to theoretical considerations of a more analytical nature \cite
{Nieto-RodriguezLopez, Nisbet-Gurney,  McKane-Newman}. According to \cite{Hidalgo-Munoz}, we can
determine that some interesting phenomena are  ruled out from our model because of the lack of foci
fixed points (again, recalling lemma \ref{lem:02}).  Namely, it is not possible to find Stochastic
Amplification of Fluctuations (SAF), a phenomenon that offers  a possible explanation to emergent
quasi-oscillations observed in fields as diverse as ecology  \cite{McKane-Newman}, epidemiology
\cite{Alonso-Pascual}, or brain dynamics \cite{Hidalgo-Munoz, Wallace-Cowan}.  In SAF the spectrum
of the noise would present a prominent peak corresponding to these quasi-oscillations. Opposed  to
this, in our study case the power spectrum of intrinsic noise must present a monotonous decay
proportional  to $1/\omega^2$, and no outstanding peaks. We confirmed this result for many sets of
parameters with agent-based  simulations, as suggested above, finding no interesting features in the
spectra, in agreement with the SAF theory.

        SAF, if present, would be self evident in agent-based simulations. It could also be noted in
series of  real data if reports of language usage over time with enough precision were available.
SAF is usually associated  with adaptive reacting forces such as prey-predator or 
activator-repressor dynamics. Thus, reports of emergent  oscillations in language dynamics could 
warn us ofthe presence of such forces driving language competition and  serve for further, necessary
refinement of the model. Numerical evidence shows that periodic solutions may appear  if the status
$s$ of the languages were allowed to change over time, which is a rather realistic extension of the
model. Also, if there would exist models of language coexistence that presented SAF in a natural
way, the observation  (or the not observation) of this phenomenon in real data could help us
determine which one is closer to reality.

    \section*{Acknowledgement} 

      This research has been partially supported by Ministerio de Econom\'i­a y Competititvidad,
project MTM 2010-15314, and Xunta de Galicia and FEDER. The authors wish also to acknowledge  the
contribution of Beatriz M\'aquez, whose contribution in the preparation of some of the  material was
very helpful.

  \begin{figure}
    \includegraphics[width=\textwidth]{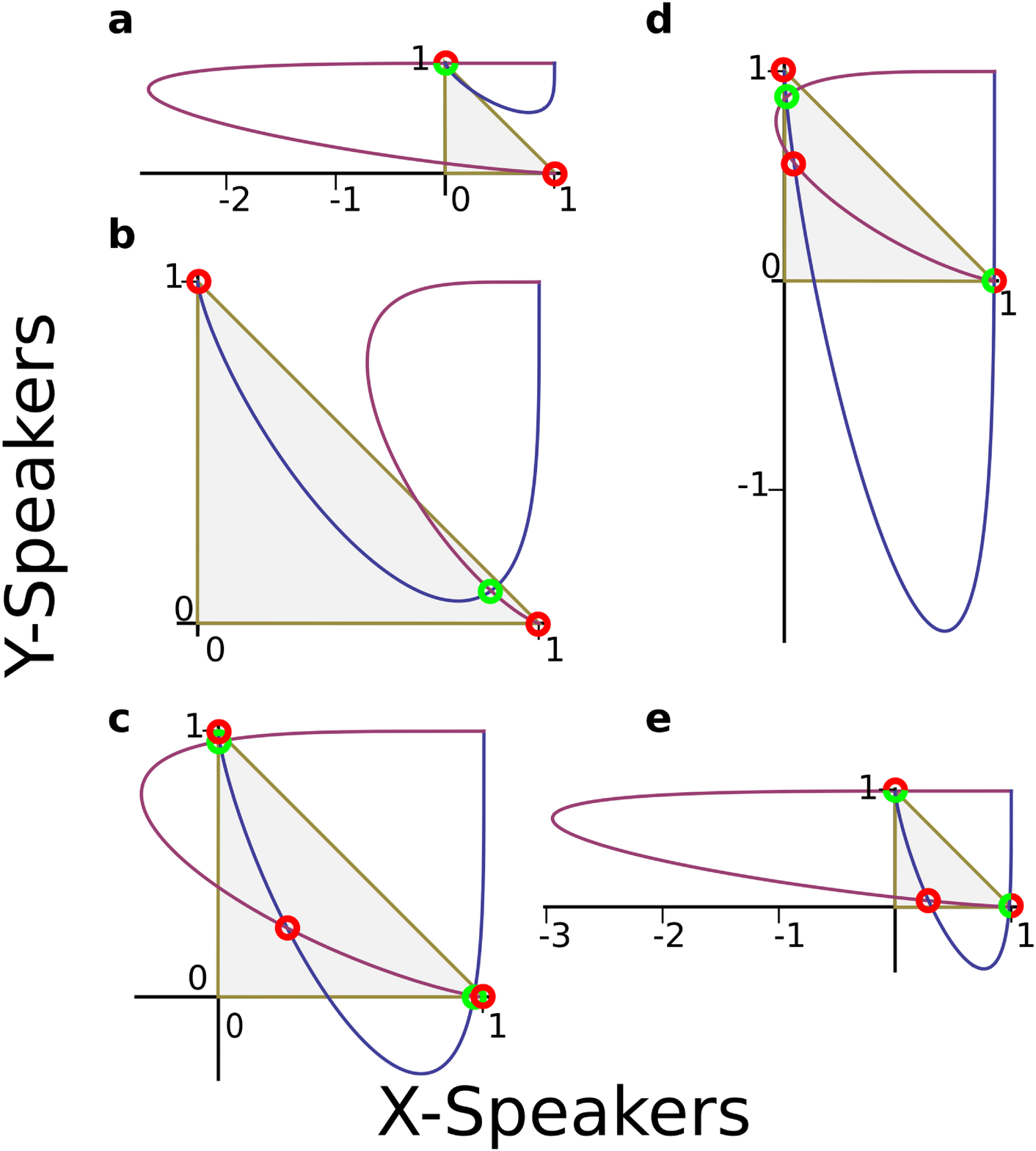}
    \caption{\textbf{Non-trivial nullclines of the system for different values of $k$ and $s$ and
fixed $a=1.31$. } Five pairs of parameters $(k,s)$ were used to plot the non-trivial branches of the
nullclines. The values were  chosen to compare our results with those from \cite{Mira-Nieto}, thus:
\textbf{a} $k = 0.65$, $s = 0.80$,  \textbf{b} $k = 0.20$, $s = 0.40$, \textbf{c} $k = 0.65$, $s =
0.50$, \textbf{d} $k = 0.75$, $s = 0.35$, and  \textbf{e} $k = 0.75$, $s = 0.35$. Stable fixed
points are indicated in red and saddle points in green. Contrary  to what was interpreted in \cite
{Mira-Nieto}, cases \textbf{a} and \textbf{b} on one side and \textbf{c}, \textbf{d},  \textbf{e} on
the other side are equivalent to each other. The five differentiated cases are obvious, though, if
we  attend to the sizes of certain basins of attraction: some of them can hardly be detected with
numerical systems. We  can appreciate in the different figures how important features, like
crossings with the boundary evolve for the different  parameters. This can be better grasped in
\cite{Maquez}. }
    \label{fig:aFixed}
  \end{figure}
  
  \begin{figure}
     \includegraphics[width=\textwidth]{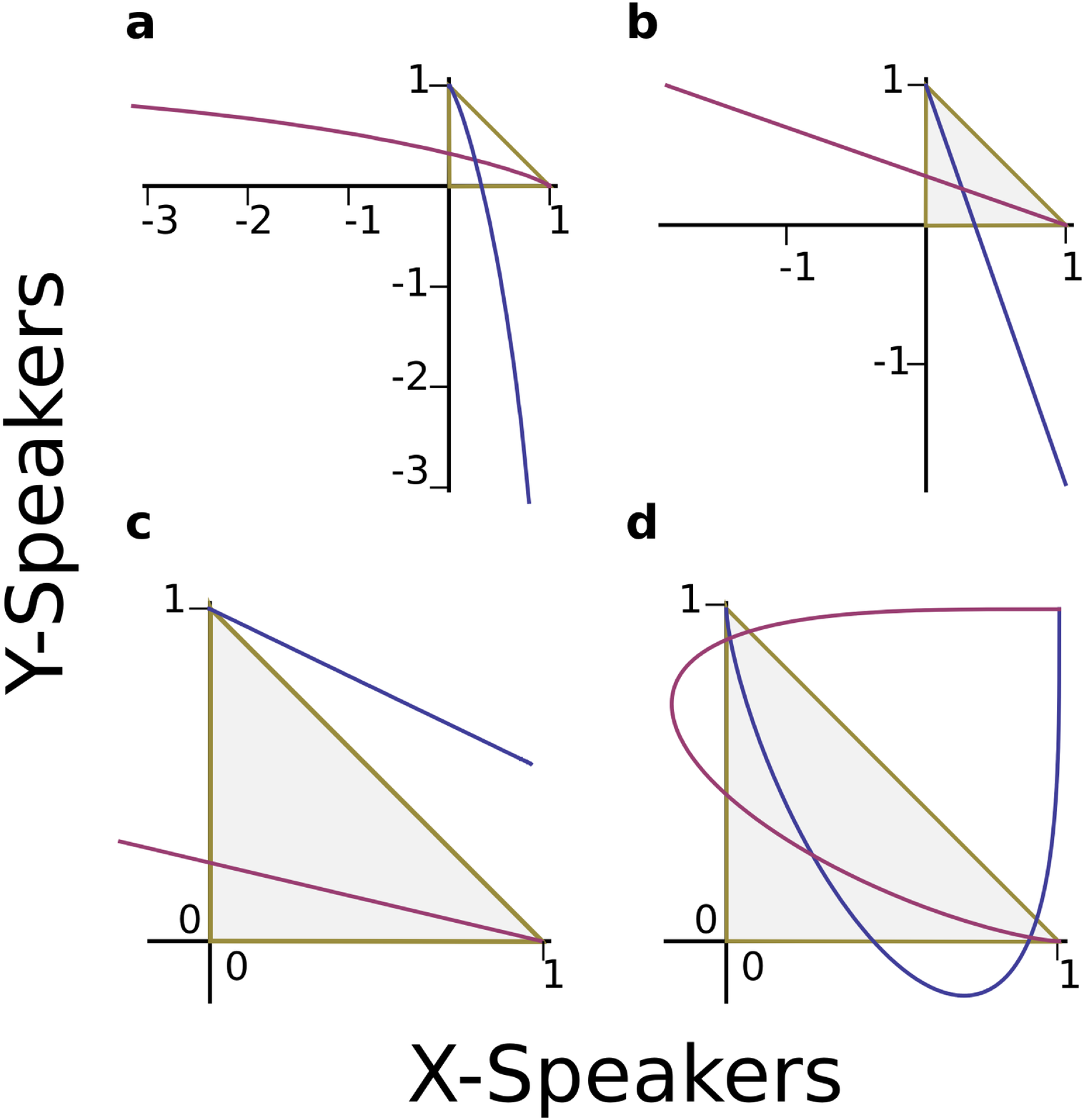}
     \caption{\textbf{Non-trivial nullclines of the system for different values of $a$. } Fixed
$k=0.65$ and $s=0.50$ is taken except in panel \textbf{c}. This figure illustrates how the non-
trivial branches of the nullclines suffer a sudden change as the parameter $a$ starting from $a<1$
increases through $a=1$ and to $a>1$. \textbf{a} For $a<1$ the nullclines must always cross inside
$A$. \textbf{b} $a=1$ is the only case when the nullclines are straight lines that go through
$(0,1)$ and $(1,0)$ respectively. As straight lines, depending on their slopes they might cross
inside $A$, but this is not necessarily the case (\textbf{c} $a=1$, $k=0.3$, and $s=0.75$).
\textbf{d} For $a>1$ the shape of the nullclines becomes more complicated and the possible
crossings must be carefully addressed. We saw the many possibilities in fig. \ref{fig:aFixed}. }
    \label{fig:aVary}
  \end{figure}
  
  \begin{figure}
    \includegraphics[width=\textwidth]{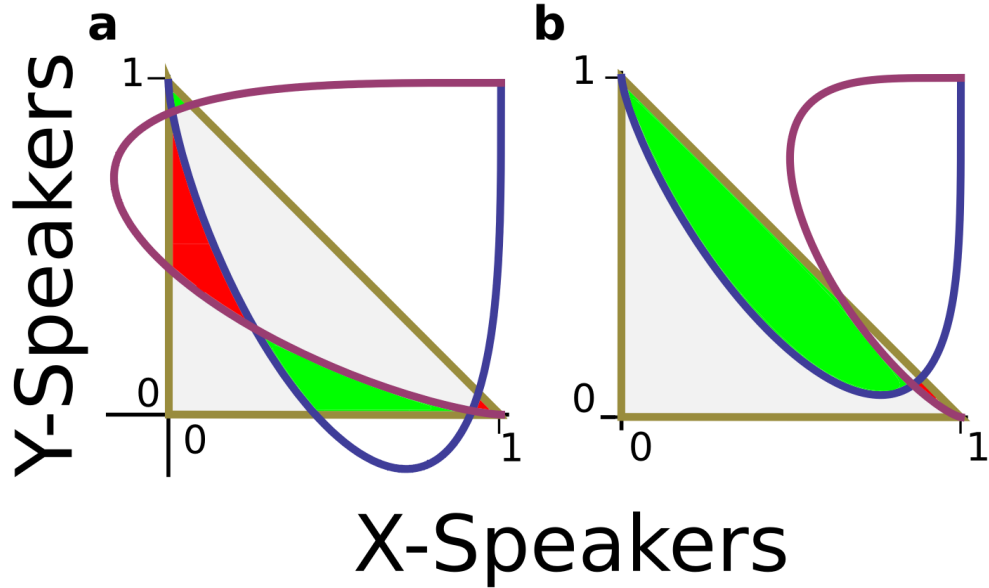}
    \caption{\textbf{Positive invariant regions $R_1$ and $R_2$ on the $x-y$. } We identify $R_1$,
where $F_x$ is negative and $F_y$ is always positive (green in the figure); and $R_2$, where $F_x$
is positive and $F_y$ is negative (red regions). \textbf{a} $a=1.40$, $s=0.50$, $k=0.65$. Three
equilibrium points exist inside $A$ and regions $A_3$ and $A_4$ are non-empty. Both $R_1$ and $R_2$
present two connected components ($A_1$ and $A_3$, and $A_2$ and $A_4$ respectively) which are, each
of them, positive invariant. This means that the dynamics do not exit any of these regions once they
enter: they must tend to a stable fixed point in their boundary. Thus, we see how points inside
$A_1$ are taken to $P_y$ and points inside $A_2$ are taken towards $P_x$. The only possibility for
regions $A_3$ and $A_4$ is that it exists another stable fixed points exactly in the joint between
the two of them. \textbf{b} $a=1.31$, $k = 0.20$, $s = 0.40$. In this case only one equilibrium
point exists inside $A$. Regions $A_3$ and $A_4$ are empty, but the same as before applies to
regions $A_1$ and $A_2$: points in their interior must be driven towards $P_y$ and $P_x$
respectively. }
    \label{fig:regions}
  \end{figure}
  
  \begin{figure}
    \includegraphics[width=\textwidth]{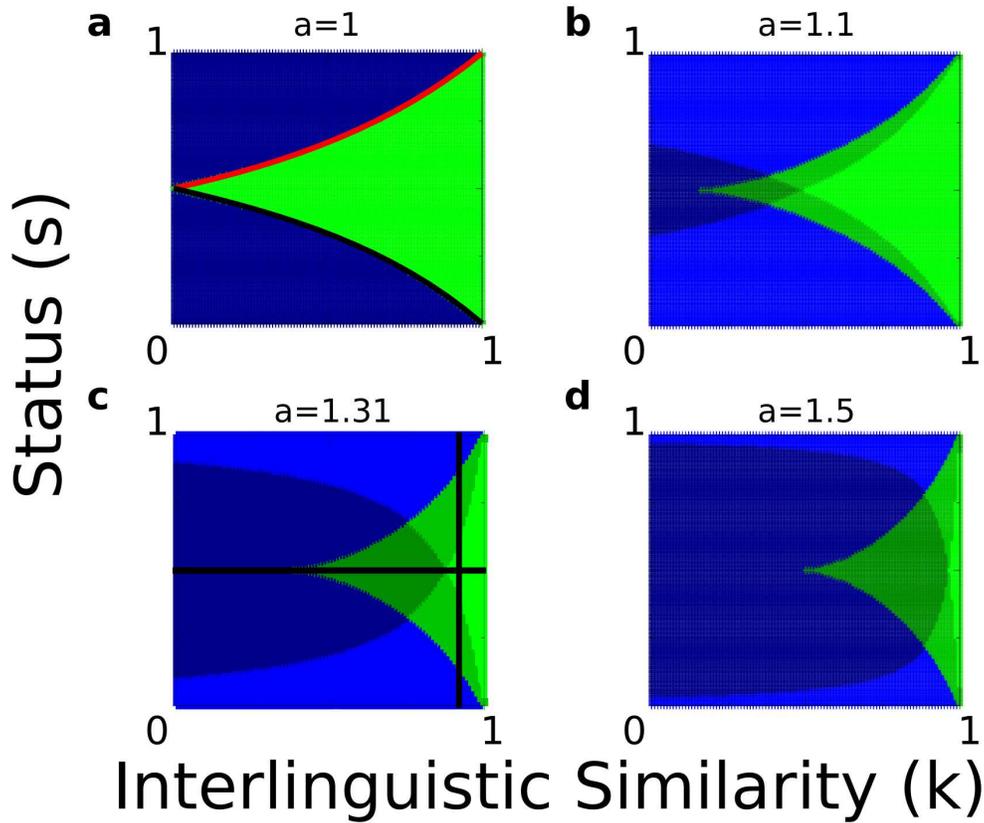}
    \caption{\textbf{Regions of the $k-s$ space with different number of equilibrium points. }
\textbf{a} a=1: It is possible to compute analytically the curves $s_{P_y}(k_{P_y})$ (thick black
line) and $s_{P_x}(k_{P_x})$ (thick red), which determine the frontiers above (below) which $P_y$
(respectively $P_x$) are unstable. The green area in which both $P_y$ and $P_x$ are unstable need
the existence of an equilibrium point inside $A$ which is stable to attract the dynamics. The blue
regions correspond to values of the parameters where either $P_x$ or $P_y$ are stable, but only one
of them; meaning, in terms of competing languages, that one tongue must extinguish the other
whatever the initial conditions. While the curves $s_{P_y}(k_{P_y})$ and $s_{P_x}(k_{P_x})$ were
computed analytically, the colored regions were found out through computer simulations: a point of
the $k-s$ space would be painted in blue if either of the attractors $P_x$ or $P_y$ were found after
evolving the system a time large enough, and it would be painted in green whenever the simulations
did not converge towards $P_x$ nor $P_y$ in a similar umber of iterations. Numerical and analytical
results agree. \textbf{b} $a=1.1$, \textbf{c} $a=1.31$, \textbf{d} $a=1.5$: In either of these cases
blue regions indicate that only the attractors $P_x$ and/or $P_y$ have been detected: both of them
were detected in dark blue regions while only one of them was detected in lighter blue regions.
Green regions indicate now that an attractor $P^*$ interior to $A$ has been detected: the darkest
area correspond with $P^*$ being detected along $P_x$ and $P_y$ and the lighter areas means that one
or two of $P_x$ and $P_y$ have not been detected. All these attractors were found numerically. The
different shades of blue and green aim at demonstrating how the stability of the fixed points $P_x$
and $P_y$ depending upon $k$ and $s$ has been misread in previous studies of the model 
\cite{Mira-Nieto}. The outcome of the simulations are consistent with the analytical results, but 
they must be taken as strict approximations as long as they remain numerical conclusions. We found 
an interesting case in panel \textbf{c} that corresponds to $a=1.31$: the value previously used in 
the literature to fit real data. Thick black lines indicate the transversal sections along which 
bifurcation diagrams are taken in fig. \ref{fig:bifurcation}. }
    \label{fig:phases}
  \end{figure}
  
  \begin{figure}
    \includegraphics[width=\textwidth]{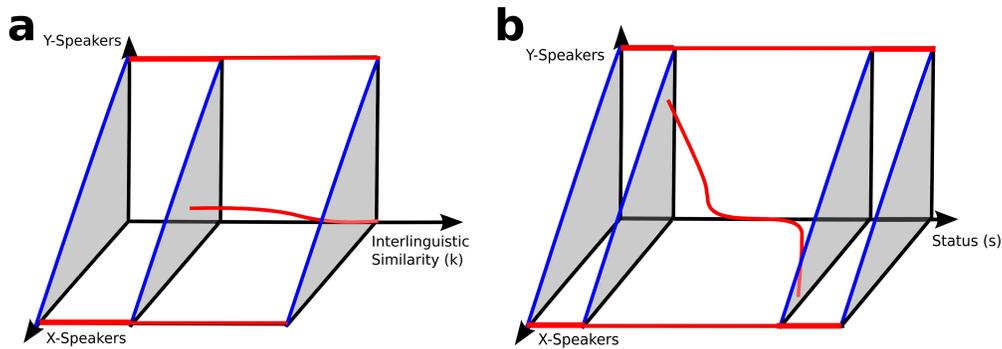}
    \caption{\textbf{Bifurcation diagrams. } Qualitative bifurcation diagrams are shown for fixed
$a=1.31$ with: \textbf{a} varying $k$ and fixed $s=0.5$ and \textbf{b} varying $s$ and arbitrary $k$
close to $0.8$. Beginning with low $k$, when varying this parameter while holding $s$ fixed we
depart from a situation in which $P_x$ and $P_y$ are stable and only one extra equilibrium point
exists in the interior of $A$, a saddle point. Above a certain value of $k$ a new stable fixed point
$P^*$ comes into existence. For varying $s$ and fixed $k$, if $k$ is large enough--as it is the case
in this example--they exist two values of $s$ between which the third stable equilibrium $P^*$ point
exists. The trajectory of $P^*$ through the $x-y$ space is qualitative in both cases. }
    \label{fig:bifurcation}
  \end{figure}

\end{document}